\documentclass{llncs}

\usepackage{amsmath}
\usepackage{graphicx}
\usepackage{algorithm}
\usepackage{algorithmic}

\newtheorem{observation}[theorem]{Observation}

\newcommand{\abs}[1]{\left\vert#1\right\vert}

\newcommand{\eps}{\varepsilon}

\newcommand{\C}{\mathcal{C}}
\renewcommand{\S}{\mathcal{S}}
\newcommand{\degout}{\operatorname{degout}}
\newcommand{\tdeg}{\widetilde{\deg}}
\newcommand{\tdegout}{\widetilde{\degout}}
\newcommand{\tcost}{\widetilde{\cost}}
\newcommand{\eqdef}{\stackrel{\Delta}{=}}

\newcommand{\tC}{\tilde{C}}
\newcommand{\tcC}{\tilde{\cal C}}
\newcommand{\Vcostly}{V_{\mathrm{costly}}}
\newcommand{\vcostly}{v_{\mathrm{costly}}}
\newcommand{\dcost}{\operatorname{cost}^{\cal D}}

\def\S{{\mathcal S}}

\def\poly{\operatorname{poly}}
\def\one{{\bf 1}}

\def\LP{\operatorname{LP}}

\def\smallnegspace{\vspace{-1ex}}

\def\utility{{\operatorname{utility}}}
\def\R{{\bf R}}
\def\C{{\mathcal C}}
\def\cost{\operatorname{cost}}
\def\OPT{{\operatorname{OPT}}}

\def\negspace{\vspace*{-2ex}}
\def\bignegspace{\vspace*{-5ex}}

\def\smallnegspace{\vspace*{-1ex}}

\begin{document}
\title{No need to choose: How to get both  a PTAS and Sublinear Query Complexity}
\titlerunning{PTAS with Sublinear Query Complexity}
\author{Nir Ailon\inst{1} \and Zohar Karnin\inst{2}}
\institute{Technion IIT and Yahoo! Research, Haifa, Israel
\email{nailon@cs.technion.ac.il}
\and Yahoo! Research, Haifa, Israel
\email{zkarnin@yahoo-inc.com}
}


\maketitle
\negspace
\begin{abstract}

We revisit various PTAS's (Polynomial Time Approximation Schemes) for
minimization versions of dense problems, and show that they can be performed
with {sublinear query complexity}.
This means that not only do we obtain a $(1+\eps)$-approximation to the
NP-Hard problems in polynomial time, but also avoid reading the entire
input.
This setting is particularly advantageous when the price of reading parts of
the input is high, as is the case, for examples, where humans provide the
input.

Trading off query complexity with approximation is the raison d'etre of the
field of learning theory, and of the ERM (Empirical Risk Minimization)
setting in particular.  A typical ERM result, however,  does not deal with
computational complexity.
We discuss two particular problems for which (a)  it has already been shown
that
 sublinear querying is sufficient for obtaining a $(1+\eps)$-approximation
using unlimited computational power (an ERM result), and (b) with full
access to input, we could
get a $(1+\eps)$-approximation in polynomial time (a PTAS).  Here we show
that neither benefit
need be sacrificed.  We get a PTAS with efficient query complexity.

The first problem is known as Minimal Feedback Arc-Set in Tournaments
(MFAST). A PTAS has been discovered by Schudy and Mathieu, and an ERM result
by Ailon.  The second is $k$-Correlation Clustering ($k$-CC).  A PTAS has
been discovered by Giotis and Guruswami, and an ERM result by Ailon and
Begleiter.

Two techniques are developed.
The first  solves the problem for the
low-cost case of $k$-CC (the analogous case is already known for MFAST).
This requires a careful sampling scheme
together with proof of a structural property relating costs
of vertices against the optimal sample clustering with their costs against
the full optimal clustering.
The second  addresses
 the high-cost case, by showing that a classic method by Arora et al. (2002)
for obtaining additive approximations  can be made query
efficient.  The underlying technique is  ``double sampling'':  One sample
 is amenable to exhaustive solution enumeration, but well
approximates only polynomially many  solutions (including the optimal), and
another sample cannot be used  exhaustively search solutions, but well
approximates the cost of the enitre solution space, and is used for
verification.

\end{abstract}



\section{Introduction}

We study two NP-Hard combinatorial minimization problems 
for which it is known how to get a  $(1+\eps)$-approximate
solution  under two scenarios.  In the first scenario,
the algorithm has full access to the input, and is required to
compute in polynomial time.  In the second scenario, the
algorithm has exponential computational power but is allowed
to uncover only a sublinear amount of input.  In this work
we show that no requirement needs to be sacrificed.  
In other words, we satisfy the following three requirements simultaneously:
\begin{enumerate}
\item[(R1)] A polynomial time algorithm.
\item[(R2)] A $(1+\eps)$ approximate solution.
\item[(R3)] A sublinear (in input size) query complexity.
\end{enumerate}

The first problem is known as
$k$-Correlation Clustering ($k$-CC).
Given an undirected graph $G=(V, E)$, the objective is to find a decomposition
of $V$ into $k$ (possibly empty) disjoint subsets (clusters) $C_1,\dots, C_k$ so 
that the symmetric difference between $E$ and 
the set $\{(u,v): \exists i \mbox{ s.t. } \{u,v\}\subseteq C_i\}$ is minimized.
The second problem is the Minimum Feedback Arc-set in Tournaments (MFAST). In this problem, given a tournament $G=(V, A)$, the objective is to write its vertices in a sequence from left to right so that the number of edges pointing to the left (\emph{backward edges}) is minimized.\footnote{By tournament we mean
that for all distinct $u,v\in V$, either $(u,v)\in A$ or $(v,u)\in A$.}
Requirements (R1) and (R2) are achieved by Giotis et. al in \cite{GiotisGuruswami06} for $k$-CC and by Kenyon-Mathieu et. al in \cite{MatSch2007} for MFAST.  Requirements (R2) and (R3) were achieved very recently by Ailon et. al 
in \cite{AilonB12} for $k$-CC  and by Ailon in \cite{Ailon11:active} of MFAST.  In this work we obtain (R1)+ (R2)+(R3) for both problems.  Our result uses components
from the aforemention citations, together with new ideas required for obtaining our strong guarantees.  

\negspace
\subsection{Previous Work and Our Contribution}

In the world of combinatorial approximations, 
Correlation Clustering  (CC) (also known as cluster editing)
has been  defined by Blum et al. \cite{BBC04}. In the
original version there was no bound on the number
$k$ of clusters.  Correlation clustering 
is max-SNP-Hard \cite{CharikarW04}
but admits constant factor polynomial times approximations (e.g. \cite{CharikarW04,Ailon:2008:AII}).
Maximization versions have also been considered \cite{DBLP:conf/soda/Swamy04}.  In this work we concentrate
on the minimization problem only, which is more difficult for
the purpose of obtaining a PTAS.  The $k$-correlation
clustering ($k$-CC), in which the number of output clusters
is bounded by $k$, is also NP-Hard but admits a PTAS
\cite{GiotisGuruswami06} running in time $n^{O(9^k/\eps^2)}\log n$. 
There is a natural machine learning
theoretical interpretation CC:  The instance space is identified
with the space of element pairs, and each edge (resp. non-edge) in $G$ is a label stipulating equivalence (resp. non-equivalence) of the corresponding pair.
The CC cost minimizes the \emph{risk}, defined as the number
of pairs of elements on which the solution disagrees with.
Roughly speaking, an algorithm attempting to minimize the risk  by, instead, minimizing an  estimator thereof obtained by sampling labels, is an \emph{Empirical Risk Minimization} (ERM) algorithm.  An ERM algorithm need not be constrained
by computational restrictions, and should be thought of as a \emph{information theoretical}, not computational result.
It should be noted that machine learning clustering theoreticians and practitioners have been studying how to use correlation clustering type labels in conjunctions with more traditional
geometric clustering approaches (e.g. $k$-means - see Basu's thesis \cite{basu05} and references therein).  Such labels are expensive because they require solicitation from humans.  Minimizing  query complexity is hence  important.

From a combinatorial  optimization point of interest, 
MFAST is NP-Hard \cite{Alon06} but admits a PTAS \cite{MatSch2007} (see references therein for a more elaborate history of this important problem).  The problem also
has a machine learning theoretical interpretation, if we
think of the directionality of the edge connecting $u$ and $v$ in the tournament $T$  as a label.  An ERM result has been
obtained by Ailon \cite{Ailon11:active} very recently. 
Interestingly, although  Ailon's algorithm is not computationally efficient, it relies quite heavily
on the ideas used in the PTAS \cite{MatSch2007}.

In this work we obtain requirements (R1),(R2) and (R3) simultaneously,
for both $k$-CC and MFAST.


\negspace
\section{Notations}
\negspace
For a natural number $n$ we denote by $[n]$ the set of integers $\{1,\ldots,n\}$.
Let $V$ denote a ground set of $n$ elements.  
In the $k$-CC problem, $V$ is endowed with an undirected graph $G=(V, E)$.
A solution to the problem is given as a clustering $\C = \{C_1,\dots, C_k\}$ of $V$ into $k$ disjoint
parts. We define $\equiv_\C$ to be the equivalence relation in which $C_1,\dots, C_k$ are the equivalence classes.
 Equivalently, we view a solution as an undirected graph $G(\C) = (V, E(\C))$ in which  $(u,v)\in E(\C)$
if and only if $u \equiv_\C v$.  The cost $\cost_G(\C)$ of a solution $\C$ is the cardinality of the symmetric difference between the sets $E$ and $E(\C)$.  When the input is
clear from the context, we will simply write $\cost(\C)$.

In the MFAST problem, $V$ is endowed with a tournament graph $T=(V, A)$.\footnote{A \emph{tournament} means
that exactly one of $(u,v)$ or $(v,u)$ are in $A$ for all $u\neq v$.}
A solution is an injective function $\pi : V\mapsto [n]$ (a permutation).  We define $\prec_\pi$ to denote
the induced order relation, namely: $u \prec_\pi v$ if and only if $\pi(u) < \pi(v)$.  Equivalently, a solution
can be viewed as a tournament $T(\pi) = (V, A(\pi))$, where $(u,v)\in A(\pi)$ if and only if $u \prec_\pi v$.  The
loss $\cost_T(\pi)$ of a solution is the number of edges $(u,v)\in T(\pi)$ such that $(v,u)\in A$. In words, a unit cost
is incurred for each inverted edge.  When the input $T$ is clear from the context, we may simply write $\cost(\pi)$.

\negspace
\section{Statement of Results and Method Overview}\label{sec:overview}
\negspace
As in \cite{GiotisGuruswami06,MatSch2007}, 
our query efficient PTAS for both $k$-CC and MFAST,  distinguishes between a high cost case and a low cost case.
In MFAST, \emph{high cost}
means that the optimal solution has cost  at least $P(\eps) n^2$, where $P(\eps)= \Theta(\eps^2 )$.
In $k$-CC, high cost means that the optimal solution has
cost at least $Q(\eps, k) n^2$, where $Q(\eps, k) = \Theta(\eps^6/k^{18})$.

In the low cost case, the problem has been solved
for MFAST by Ailon \cite{Ailon11:active}.  
There it is shown
that $O(n\eps^{-4}\log^4 n)$ edges from $T$ are sufficient
for finding a $(1+\eps)$-approximate solution,
in $\poly(n, \eps^{-1})$ time.
We refer the reader to 
\cite{Ailon11:active} for the details.
As for the low cost case for $k$-CC,  we show a PTAS with
$o(n^2)$ query complexity in 
 in Section~\ref{sec:low}.  The main idea of the algorithm
is similar to that in \cite{GiotisGuruswami06}, but defers
in a significant way.
 Roughly speaking, both algorithms choose a sample of vertices 
and enumerate over $k$-clusterings of the sample, while trying to compute optimal big clusters from the sample.  In \cite{GiotisGuruswami06}, for each such choice of sample $k$-clustering, a clustering of $V$ is chosen, and recursion is executed on the union of small clusters.  
Here, we use the sample \emph{in vitro} to learn a strong structural property of any optimal solution for the entire input.  In particular, we
don't need to return from a recursion to perform this learning.

For the high cost case we invoke an algorithm
giving an additive $\eps P(\eps)n^2$ (resp. $\eps Q(\eps, k)n^2$) approximation for MFAST (resp.  for $k$-CC).
To that end, we use a  standard LP based technique \cite{DBLP:journals/mp/AroraFK02}, together with another
double sampling trick  necessary for query efficiency, which we describe in  Section~\ref{sec:high} for the MFAST case (the $k$-CC case is easier).  The main result there is as follows:
\smallnegspace
\begin{theorem}\label{thm:highcost}
There exists a polynomial (in $n$) time algorithm
for obtaining an additive $\eps P(\eps)n^2$ (resp. $\eps Q(\eps, k) n^2)$) approximation for MFAST (resp. for $k$-CC).
The algorithm queries  $O(\eps^{-2}P^{-2}(\eps)n\log n)$
(resp. $O(\eps^{-2}Q^{-2}(\eps, k) n^2)$) input edges
and runs in time $n^{O(\eps^{-2}P^{-2}(\eps)\log P(\eps))}$ 
(resp. $n^{O(\eps^{-2}Q^{-2}(\eps, k) \log k)}$).
\end{theorem}

\smallnegspace
In order to know whether we are at all in the high cost case,
we apply the additive approximation algorithm in any case, and
approximate the cost of the returned solution to within an additive error
of $\Theta( P(\eps) n^2)$ (resp.  $\Theta(Q(\eps, k) n^2)$).  This estimation can clearly be done, with success probability at least $1-n^{-10}$, by sampling at most $O(P^{-2}(\eps)\log n)$
(resp. $O(Q^{-2}(\eps,k)\log n)$ ) edges, by standard measure concentration arguments. \footnote{Note that the algorithm in \cite{Ailon11:active}
relies on a divide and conquer recursive strategy, in which 
the high cost algorithm and test must be implemented at
each recursion node.  This also holds for our $k$-CC algorithm, which identifies large clusters and then recurses on small ones.  The recursive calls must solve and test for the high cost case as well.} This bound is overwhelmed by
the bounds of Theorem~\ref{thm:highcost}.
Our
 main results are summarized as follows.
\smallnegspace
\begin{theorem} \label{thm:PTAS low cost kcc}
There exists a PTAS for $k$-CC
running in time  $n^{O(\eps^{-14}k^{36}\log k)}$ 
and requiring at most $O(\eps^{-14}k^{36}n\log n)$ edge queries. 
With probability at least $1-n^{-3}$, 
it outputs a clustering $\tcC$ with $\cost(\tcC) \leq\cost(\C^*)(1+\eps)$, where $\C^*$ is an optimal solution.
\end{theorem}
\negspace\begin{theorem} \label{thm:PTAS low cost mfast}
There exists a PTAS for MFAST
running in time  $n^{O(\eps^{-6})}$ and requiring at most $O(\eps^{-6}n\log n + \eps^{-4}n\log^4 n)$ edge queries. 
With probability at least $1-n^{-3}$, it outputs a permutation $\sigma$ with $\cost(\sigma) \leq\cost(\pi^*)(1+\eps)$, where $\pi^*$ is an optimal solution.
\end{theorem}
\smallnegspace
Note that the running times are overwhelemed by the high
cost case in both, and the query complexity is overwhelemed
by the high cost case in Theorem~\ref{thm:PTAS low cost kcc}.
We also note that we did not make a real effort to optimize the constants, including the exponents of $k,\eps$.


\section{Query Efficient PTAS for Low Cost in  $k$-CC }\label{sec:low}
\negspace
We study the low cost case of $k$-CC on input $G=(V, E)$, and analyze an algorithm satisfying (R1)+(R2)+(R3).  
We need two  ingredients.  In Section~\ref{sec:low1} we approximate
the contribution of a single node $v$ to the cost of any solution
identical to the optimal solution except (maybe) for a change
in the cluster to which $v$ belongs.
In Section~\ref{sec:low2} we achieve the PTAS, using a strategy similar to that of Giotis et al. in \cite{GiotisGuruswami06}:  Identification of the large clusters in the optimal solution and recursion on the remainder. Note that the algorithm of \cite{GiotisGuruswami06} does not satisfy (R3), hence ours makes better use of the queried information.
\negspace
\subsection{An additive approximation of vertex costs}\label{sec:low1}

A major component in our PTAS for $k$-clustering is an additive approximation for the contribution of each vertex to the cost of the clustering. We  start by formally defining this contribution, and then present Algorithm~\ref{alg:cost_add_aprx} and its analysis.

\begin{definition}
Let ${\cal C^*}=\{ C_1^*,\ldots,C_k^*\}$ be an optimal $k$-clustering, and assume its cost is $\gamma n^2$ for some $\gamma\geq 0$.
For $v\in V$  let $j^*(v)$ be defined as the unique index such that $v\in C^*_{j(v)}$.  Let $C^*(v) = C^*_{j^*(v)}$.
Let $\one_{u+v}$ be an indicator variable for the predicate $(u,v)\in E$, and similarly define the complement $\one_{u-v} = 1 - \one_{u+v}$.
 Let $ \deg_+(v,j) = \sum_{u \in C^*_j \setminus \{v\}} \one_{u+v}$,  
$\deg_-(v,j) = \sum_{u \in C^*_j\setminus\{v\}} \one_{u-v}$, $\degout_+(v,j) = \sum_{u \in (V\setminus C^*_j)\setminus\{v\}} \one_{u+v}$, and $\degout_-(v,j) = \sum_{u \in (V\setminus C^*_j)\setminus\{v\}} \one_{u-v}$.
 Let $\cost^*(v) = \sum_{ u \in C^*(v)\setminus\{v\}} 1_{u-v} + \sum_{u \notin C^*(v)} 1_{u+v} $. Notice that $\cost^*(v) = \deg_-(v,j(v)) +  \degout_+(v,j(v))$ and  $\cost(\C^*)=\frac{1}{2} \sum_v \cost^*(v)$.
 For any $j \in [k]$, let $\cost^*(v,j) \eqdef \deg_-(v,j) + \degout_+(v,j)$. That is, $\cost^*(v,j)$ is the contribution of the vertex $v$ to the cost of the clustering that is identical to $\C^*$, except the location of $v$, which is reset to $C^*_j$.
\end{definition}

\begin{algorithm} 
\caption{Additive approximation of $\cost^*$}
\label{alg:cost_add_aprx}
\emph{Input}: A graph $G=(V,E)$, a parameter $\beta>0$ and integer $k>1$\\
\emph{Output}: $\forall v \in V$, and $j\in[k]$, an estimation $\tcost(v,j)$ to $\cost(v,j)$
\newline
\\
Choose $S=(v_1,\ldots,v_t)$, where $t= c \log(n) \beta^{-9}$ ($c$ is some sufficiently large universal constant), be a multiset of i.i.d. uniformly randomly chosen vertices from $V$.

Let $\tilde{S}_1,\ldots,\tilde{S}_k$ be an optimal $k$-clustering for the reduced problem $(S, E_{|S})$ (where $E_{|S} = E \cap (S\times S)$), where the solution is found using exhaustive search.

For any $v\in V$, $j \in [k]$, let:
$ \tdeg_-(v,j) \eqdef \sum_{v\in \tilde S_j\setminus \{v\}} \one_{u-v}$ and
$\tdegout_+(v,j) \eqdef \sum_{v\in (S \setminus \tilde S_j) \setminus \{v\}} \one_{u+v}$.
(The summations count elements of $S$ with multiplicities.)

Output for every $v\in V$, $j \in [k]$ the estimation: $$\tcost(v,j)\eqdef \frac n {|S|}  \left (\tdeg_-(v,j)+\tdegout_+(v,j)\right )\ .$$
\negspace
\end{algorithm}

\noindent
The rest of this section proves the following
guarantee of Algorithm~\ref{alg:cost_add_aprx}.
\begin{theorem} \label{thm:cost_add_apx}
Fix $\beta>0$, to be passed as paramater to Algorithm~\ref{alg:cost_add_aprx}.
There exist some universal constant $c$ such that if $\gamma \leq c \beta^6$ then for all $v \in V$, $j \in [k]$ it holds that for the output of the algorithm, 
after possibly renaming the optimal clusters $\{C^*_1,\dots, C^*_k\}$,
$\abs{ \tcost(v,j) - \cost^*(v,j)  } < \beta n.$
For any input, Algorithm~\ref{alg:cost_add_aprx} will run in $n^{O(\beta^{-9}\log k)}$ time
and will require at most $O(n\log(n)\beta^{-9})$ edge queries.
\end{theorem}

The claim regarding the time and query complexity of the algorithm are trivial.  Indeed, the time is dominated by exhaustively searching the space of $k$-clusterings of the sample $S$ in the algorithm. We focus on proving the correctness.
We need some more definitions.

\begin{definition}
Let $u,v \in V$,  $S$ a multi-subset of $V, j\in [k]$ and $\delta > 0$.
Let
$\deg_+^S(v,j) = \sum_{u \in (C^*_j\cap S) \setminus \{v\}} \one_{u+v}$, $\deg_-^S(v,j) = \sum_{u \in (C^*_j\cap S)\setminus\{v\}} \one_{u-v}$, $\degout_+^S(v,j) = \sum_{u \in (S\setminus C^*_j)\setminus\{v\}} \one_{u+v}$ and 
$\degout_-(v,j) = \sum_{u \in (S\setminus C^*_j)\setminus\{v\}} \one_{u-v}$,
where the summations take multiplicities in $S$ into account.
Let $\cost^{*S}(v) \eqdef \deg_-^S(v,j^*(v))+\degout_+^S(v,j^*(v))$. 
\end{definition}

In what follows, set $\delta=\Theta(\beta^3)$. 
Define $\S$ as the partition of $S$  (from Algorithm~\ref{alg:cost_add_aprx}) induced by $\C^*$. That is
$\S=\{S_1,\ldots,S_k\}$ where $S_j=C^*_j \cap S$.

\begin{lemma} \label{lem:deg in S and C equal}
With probability at least $1-n^{-10}$, 
for all $v \in V$ and $j\in[k]$,
\begin{eqnarray*}
\max\{& & |\deg_+(v,j)/n - \deg_+^S(v,j)/|S||, 
|\deg_-(v,j)/n - \deg_-^S(v,j)/|S||, \\
& & |\degout_+(v,j)/n - \degout_+^S(v,j)/|S||,
|\degout_-(v,j)/n - \degout_-^S(v,j)/|S||\} = O(\delta) \ .
\end{eqnarray*}
\end{lemma}
The simple proof is deferred to Appendix~\ref{sec:proof:lem:deg in S and C equal}.
From Lemma~\ref{lem:deg in S and C equal},
\begin{lemma}\label{lem:cost_induced_partition_on_S}
Assume $\gamma = o(\delta)$.
With probability $1-n^{-10}$, the cost of the partition $\S$ on the graph $G|_S = (S, E|_S)$ is at most $O(\delta|S|^2)$.
\end{lemma}

In the following lemma we show that any pair of clusterings that are close w.r.t.\ to their edges are also close w.r.t.\ their vertices.

\begin{lemma} \label{lem:close in V}
Let $\S,\tilde{\S}$ be two $k$-clusterings of $S$,
and let $E(\S), E(\tilde \S)$ be their corresponding edge sets,
namely, $(u,v)\in E(\S)$ if and only if $u \equiv_\S v$,
and similarly for $\tilde S$.
Assume the size of the symmetric difference between
$E(\S)$ and $E(\tilde \S)$ is at most $\delta |S|^2$,
where $\delta < c/k^3$ and $c$ is a 
sufficiently small constant. Then for some reordering of indices, for every $j \in [k]$, 
$\max \{|S_j \setminus \tilde{S_j}|, |\tilde{S}_j \setminus S_j| \}= O(\delta^{1/3} |S|)\ .$
\end{lemma}

We will only present a main structural claim used by the proof.
The remainder of the proof will be deferred to Appendix~\ref{sec:proof:lem:close in V}.
\begin{proof}
We start with an auxilary claim showing that every cluster in $\S$ has a similar cluster in $\tilde{\S}$  (and vice versa).

\begin{claim} \label{clm:aux close in V}
Let $C$ be a cluster of $\S$. There exists some cluster $D$ in $\tilde{S}$ such that $|C \setminus D| \leq O(\delta^{1/3}|S|)$.
\end{claim}
\begin{proof}
Let $D$ be a cluster in $\tilde{S}$ that maximizes $|D \cap C|$. Let $A = D \cap C$ and let $\bar{A} = C \setminus A$.
Notice that for every pair $(u,v)\in A\times  \bar{A}$ the edge $(u,v)$ is an element of  $E({\S})\setminus E(\tilde \S)$.  
Hence,
$ |A||\bar{A}| \leq \delta |S|^2 $.
If $|C|<\delta^{1/3}|S|$ then the claim holds trivially. If $|C|\geq \delta^{1/3}|S|$, then we get: 
$ |A|(|C|-|A|) = |A||\bar{A}| \leq \delta |S|^2 \leq \delta^{1/3} |C|^2$.
A simple calculation will show that either $|A|\leq O(\delta^{1/3} |C|)$ or $|A| \geq |C|(1-O(\delta^{1/3}))$. By setting the constant $c$ to be sufficiently small we get that the first option implies $|A|<|C|/k$ which is impossible due to the fact that $A$ maximizes $|D\cap C|$ over all clusters $D$ in $\tilde \S$. definition of $A$. We conclude that $ |C \setminus D| = O(\delta^{1/3} |S|)$,  
proving the claim.
The remainder of the proof of Lemma \ref{lem:close in V} is deferred to Appendix~\ref{sec:proof:lem:close in V}.

\end{proof}
\end{proof}

Let $\tilde{\S}=\tilde{S}_1,\ldots,\tilde{S}_k$ be 
an optimal $k$-clustering of the induced input $G|_S$.
By Lemma~\ref{lem:cost_induced_partition_on_S},
we know that with probability at least $1-n^{-10}$
the cost of the solution $\tilde \S$ is at most $\delta |S|^2$.
By the triangle inequality, this implies  that the symmetric
difference between $E(\S)$ and $E(\tilde \S)$ is at most $O(\delta |S|^2)$.
Hence, we may apply Lemma~\ref{lem:close in V}
and assume that the clusters $S_1,\dots, S_k$ and $\tilde S_1,\dots \tilde S_k$ are aligned with each other.
Define:
$\tdeg_+(v,j) = \sum_{u \in \tilde S_j\setminus \{v\}} \one_{u+v}$,  $\tdeg_-(v,j) = \sum_{u \in  \tilde S_j\setminus \{v\}} \one_{u-v}$, $\tdegout_+(v,j) = \sum_{u \in (S \setminus\tilde S_j)\setminus \{v\}} \one_{u+v}$, and  
$\tdegout_-(v,j) = \sum_{u \in  (S\setminus \tilde S_j)\setminus \{v\}} \one_{u-v}$.


\begin{lemma}
With probability at least $1-n^{-8}$, for all $v\in V$ and $j \in [k]$,
$ \abs{\frac {\tdeg_+(v,j)}{|S|} - \frac{ \deg_+(v,j)}{n}} = O(\delta^{1/3}) .$
The same is true for the other `deg functions'.
\end{lemma}
\begin{proof}
By the guarantee of Lemma \ref{lem:close in V}, for all $v \in V, j \in [k]$, 
$\abs{\frac {\tdeg_+(v,j)}{|S|} - \frac {\deg^S_+(v,j)}{|S|}}=O(\delta^{1/3})$.
By the guarantee of Lemma~\ref{lem:deg in S and C equal}, we have that  for all $v\in V,j\in [k],$ 
$\abs{\frac{\deg^S_+(v,j)}{|S|} - \frac{\deg_+(v,j)}{n}}=O(\delta^{1/3})\ .$
The claim follows by union bounding and using the  triangle inequality.  This concludes the lemma's proof.
\end{proof}

\noindent
Theorem \ref{thm:cost_add_apx} is now an easy corollary.

\subsection{The PTAS}\label{sec:low2}

In this section we utilize the approximations to the costs of the vertices achieved in Algorithm~\ref{alg:cost_add_aprx} to achieve a PTAS for $k$-clustering. We note that the heart of our contribution is the previous section, and the lemmas and proofs here follow the lines of \cite{GiotisGuruswami06}. The main algorithm (Algorithm \ref{alg:main kcc}) is of course different since it utilizes the results of the previous section.

Throughout this section we will assume that the optimal clustering $\C^*$ has a cost of $\gamma n^2$ where 
$\gamma < c_1 \beta^6$, where the parameter $\beta$
will be taken as $c_2 \eps / k^2$, and $c_1, c_2$ will be sufficiently small constants so that Theorem~\ref{thm:cost_add_apx} is satisfied.



\begin{algorithm}  \label{alg:main kcc}
\caption{PTAS for $k$-CC (low cost)}
\label{alg:CC PTAS}
\emph{Input}: A graph $G=(V,E)$, an integer $k>1$ and a parameter $\eps>0$. It is assumed that the optimal $k$-CC cost of $G$ is $\gamma n^2$, where $\gamma < c_1\beta^6$ and $\beta = c_2\eps/k^3$.

\emph{Output}: A clustering $\tcC = \{ \tC_1,\ldots,\tC_k \}$ of $G$.

Run Algorithm~\ref{alg:cost_add_aprx} with inputs $G$,$k$ and $\beta$.  Obtain approximations $\tcost(v,j)$ for all $v\in V$ and $j\in[k]$. 

Create empty clusters $\hat C_1,\dots \hat C_k$. 
For all $v\in V$ add $v$ to $\hat C_i$, where $i=\mathrm{argmin}_j\{ \tcost(v,j)\}$. 

Reorder the clusters so that $|\hat C_1|\geq \ldots \geq |\hat C_k|$. Let $\ell \in [k]$ be such that $|\hat C_\ell| \geq \frac{n}{2k}$ and $|\hat C_{\ell+1}| < \frac{n}{2k}$ (if no such integer exists, set $\ell=k$).

Run the algorithm recursively on the restriction of $G$ on $W \eqdef \cup_{j > \ell} \hat{C}_j$, the integer $k-\ell$ and approximation parameter $\eps(1-1/k)$. Denote its output by $\tC_{\ell+1},\ldots,\tC_k$. 

Output $\tcC = ( \tC_1=\hat{C}_1,\ldots,\tC_{\ell}=\hat{C}_{\ell}, \tC_{\ell+1},\ldots,\tC_k)$\ .
\end{algorithm}	

The remainder of the secion is dedicated to proving Theorem~\ref{thm:PTAS low cost kcc}.  We need some lemmas.
In what follows, we assume that the invocation of
Algorithm~\ref{alg:cost_add_aprx} is successful
in the sense that  the guarantee of Theorem~\ref{thm:cost_add_apx} holds.
The following is an immediate corollary of this guarantee.
\begin{lemma} \label{lem:cost in other}
Let $v \in V$ be a vertex satisfying  $v \in \hat C_j\cap C^*_i$, where $i \neq j$. Then $\cost^*(v,j) \leq \cost^*(v)+2\beta n$\ .
\end{lemma}

Define for any $v\in V$, $\tcost(v) = \min_{j\in[k]} \tcost(v,j)$, where
$\tcost(v,j)$ is as defined in Algorithm~\ref{alg:CC PTAS}.
Define $\Vcostly = \{v\in V \; |\; \tcost(v)  \geq c_3n/k^2 \}$, where $c_3$ is some sufficiently small constant.
For any $v\in \Vcostly$, $\cost^*(v) \geq \frac 1 2 c_3 n/k^2$ due
to the guarantee of Theorem~\ref{thm:cost_add_apx}
and our choice of $\beta$.
Since (twice) the total optimal cost is bounded by that incurred by vertices in $\Vcostly$:
\begin{equation}\label{Vcostlysize1} |\Vcostly| \leq 4 \gamma n^2 / (c_3 n/k^2) \leq (4\gamma n k^2)/{c_3}\
 .\end{equation}  
 In particular, using a very crude estimate, this means
\begin{equation}\label{Vcostlysize}
 |\Vcostly| \leq  c_4 n /k  ,
\end{equation}
where $c_4$ is a constant that can be made
sufficiently small by reducing $c_1$ as necessary. 
Recall that $\hat C_1,\ldots,\hat C_\ell$ are the large clusters found by Algorithm~\ref{alg:CC PTAS}. Notice
that since there are $k$ clusters, there must be at least one cluster of size $\geq \frac{n}{2k}$ meaning that $\ell \geq 1$.
\begin{lemma}\label{lem:Cstar almost hat C}
For any $j \in [\ell]$, $C^*_j \setminus \Vcostly = \hat C_j \setminus \Vcostly$\ .
\end{lemma}
The proof is deferred to Appendix~\ref{sec:proof:lem:Cstar almost hat C} for lack of space.
The next lemma states the existence of a clustering whose large clusters are identical to those found by our algorithm
and has an almost optimal cost.
\begin{lemma} \label{lem:large can recurs}
There exist some $k$-clustering of $V$, ${\cal D}=(D_1,\ldots, D_k)$ such that for all $j \in [\ell]$, $D_j=\hat C_j$ and $\cost({\cal D}) \leq \gamma n^2 (1+\eps/k)$ 
\end{lemma}
\begin{proof}
Take $\cal D$ to be the clustering defined as follows.
For any $i \in [k]$, 
$ D_i = (C^*_i \setminus \Vcostly) \cup (\hat C_i \cap \Vcostly)$.
That is, $D_i$ is the result starting with the clustering $\cal C$ and of moving the vertices of $\Vcostly$ to the clusters according to $\hat \C = \{\hat C_1,\dots, \hat C_k\}$. 

Denote by $\dcost(v)$ the cost of a vertex $v$ w.r.t.\ the partition $\cal D$. Notice that the only edges for which the clustering $\cal D$ pays for while the clustering $\C^*$ does not must be incident to a node in $\Vcostly$. Hence,
\negspace
\begin{equation}\label{costdiffbound1} \cost({\cal D}) - \cost({\cal C^*}) \leq\sum_{v \in \Vcostly} (\dcost(v)-\cost^*(v) )\ . \end{equation}
Assume $\vcostly \in \Vcostly\cap D_j$ for some $j\in[k]$.
\smallnegspace
Clearly
\begin{equation}\label{costdiffbound} \left |\dcost(\vcostly)-\cost^*(\vcostly,j)\right| \leq{|\Vcostly|}\ , 
\end{equation}
because the only difference in such a vertex's cost can  come from edges connecting it to other vertices in $\Vcostly$. 
Now assume $\vcostly \in \Vcostly \cap C^*_i\cap D_j$ for $j \neq i$. By construction,  $\vcostly \in \hat C_j$. By Lemma \ref{lem:cost in other}, this implies $\cost^*(\vcostly,j) \leq \cost^*(\vcostly)+2\beta$. By (\ref{costdiffbound}) we conclude
\smallnegspace
\begin{eqnarray*}
  \dcost(\vcostly)-\cost^*(\vcostly) \leq\ \ \ \ \ \ \ \ \ \ \ \ \ \ \ \  \ \ \ \ \ \ \ \     \ \ \ \ \ \ \ \    \ \ \ \ \ \ \ \  \\
(\cost^*(\vcostly,j)-\cost^*(\vcostly,i))+|\Vcostly| \leq 2\beta+|\Vcostly|\ .
\end{eqnarray*}
\noindent
Plugging this into (\ref{costdiffbound1}) and using (\ref{Vcostlysize1}), we get
\negspace
$$ \cost({\cal D}) - \cost({\C^*}) \leq |\Vcostly|  \left(2\beta n+|\Vcostly|\right) \leq  \gamma n^2 \left({8\beta k^2}/{(c_3)}+ {16\gamma k^4}/{(c_3)^2}\right)\ .$$ 
The claim follows since $\gamma < c_1 \beta^6 \leq \frac{\eps}{2k} \cdot  \frac{(c_3)^2}{16 k^4} $,  $\beta < \frac{\eps}{2k} \cdot  \frac{c_3}{8k^2}$ assuming small $c_1, c_2$.

\end{proof}

\begin{proof} [of Theorem \ref{thm:PTAS low cost kcc}]
The claim regarding the query complexity is trivial given Theorem \ref{thm:cost_add_apx}. The running time is a result of the recursion formula $T(n,\eps,k)=n^{\eps^{-9}k^{27} \log(k)}+T(n,\eps(1-1/k),k-1)=n^{\eps^{-9}k^{27} \log(k)}$.
We note that in \cite{GiotisGuruswami06}, the stated running time is doubly exponential in $k$ wheras here it is singly exponential in $k$. This difference is due to a minor observation that the recursive call should be with the parameter $\eps(1-1/k)$ rather than $\eps/10$. The same minor change would result in a singly exponential dependence in $k$ in the algorithm given in \cite{GiotisGuruswami06} as well.
Let $W$ be the union of the small clusters. That is, $W=\cup_{j=\ell+1}^k \hat C_j$.
By lemma \ref{lem:large can recurs}, all of the vertex pairs that are not contained in the set $W \times W$ incur a cost in $\hat C$ identical to that in $\cal D$.
Let $d_1$ be the cost of $\cal D$ on pairs in $W \times W$ and let $d_2$ be its cost on the remaining pairs $V\times V\setminus W\times W$. Since $W$ is clustered recursively, we have that the cost of $\hat C$ is at most $d_2+d_1(1+\eps/k) \leq (d_1+d_2)(1+\eps/k)=\cost({\cal D})(1+\eps/k)$. The statement of the theorem follows.

\end{proof}


\bignegspace
\section{Query Efficient PTAS for High Cost}\label{sec:high}
\negspace

We present a query efficient PTAS for the high loss case of MFAST. The query efficient PTAS for the high loss case of $k$-CC is almost identical and is thus not presented.
We will start by describing a known PTAS ((R1)+(R2)) based on an approach given by Arora et.\ al. We then show how to add requirement (R3).  The final approach is summarized in Algorithm~\ref{alg:high cost}, found in Appendix~\ref{sec:alg:high cost}.



\negspace
\subsection{ (R1)+(R2) using a Known Additive Approximation Algorithm} \label{sec:high cost PTAS}
Let $\pi^*$ denote an optimal permutation, and let $\OPT$ denote $\cost(\pi^*) = \cost(\pi^*)$.
In the high cost  MFAST case, as explained in Section~\ref{sec:overview}, we assume $\OPT \geq \gamma n^2$, where $\gamma = \Theta(\eps^2)$. 
Instead of directly solving MFAST, we solve the \emph{bucketed} version.  This idea is not new and can be found in e.g. \cite{MatSch2007}.  
An $m$-bucket ordering $\sigma$ of $V$ is a mapping $\sigma: V\mapsto [m]$, where for each $i\in[m]$ the preimage satisfies:
 $\frac {n}{2m}  \leq |\sigma^{-1}(i)|  \leq \frac{2n}m $.
For brevity we say that $u <_\sigma v$ if $\sigma(u) < \sigma(v)$, and $u \equiv_\sigma v$ if $\sigma(u) = \sigma(v)$.
We extend the definition of $\cost(\cdot)$ to bucketed orders by defining
$\cost(\sigma) \eqdef \sum_{u <_\sigma v} \one_{(v,u) \in A}\ .$
We will also need to define: 
$\cost^{u,v}(\sigma) \eqdef \one_{u <_\sigma v}\one_{(v,u)\in A} + \one_{v <_\sigma u}\one_{(u,v)\in A} $ and 
$\cost^u(\sigma) \eqdef \frac 1 2  \sum_{v\in V} \cost^{u,v}(\sigma)$, 
 so that $\cost(\sigma) = \sum_{u\in V} \cost_T^{u}(\sigma)$.
A permutation $\pi$ extends an $m$-bucketed ordering $\sigma$ if $u<_\pi v$ whenever $u<_\sigma v$.
\begin{observation}\label{obs:bucket}\cite{MatSch2007}
For any $\pi$ extending $\sigma$, $|\cost(\pi) - \cost(\sigma)| = O(n^2/m)$, hence for
the purpose of obtaining a $(1+\eps)$-approximate solution in our case it suffices to consider $m$-bucketed orderings
with $m = \Theta( 1/(\eps \gamma))$.
\end{observation}


\noindent
Let $\sigma^*$ denote any $m$-bucketed  ordering of $V$ of which $\pi^*$ is an extension, and such that $\lfloor n/m \rfloor \leq | (\sigma^*)^{-1}(i) | \leq \lceil n/m\rceil$ for all $i\in [m]$.
The following approach has been taken in \cite{DBLP:journals/mp/AroraFK02}.  Let $S = (v_1,v_2,\dots, v_s)$ be a random series of $s=O(\log n/(\eps\gamma)^2)$ vertices in $V$,
each element chosen uniformly and independently, with repetitions.  Abusing notation, we will also think of $S$ as the series $\{v_1,\dots, v_s\}$.
  For each $m$-bucketed ordering $\sigma$ and for each $u\in V$, we make the following definitions:
$\cost^{u,S}(\sigma) \eqdef \frac {n} {2s} \sum_{i=1}^s \cost^{u,v_i}(\sigma)$ and 
$\cost^S(\sigma) \eqdef \sum_{u\in V} \cost^{u,S}(\sigma)$. 
Clearly $\cost^{S}(\sigma)$ is an unbiased estimator of $\cost(\sigma)$ over the choice of the sample $S$.
The top level of our algorithm will enumerate over all  $n^{\Theta(\log(1/(\eps \gamma))/(\eps\gamma)^2)}$ possibilities for the value of $(\sigma^*(v_1),\dots, \sigma^*(v_s))$.  From now on, we will assume the correct possibility has been chosen, so that $\sigma^*(v)$ is ``known'' for $v\in S$.  A verification
step will be used to identify the correct possibility in the end
(see Algorithm~\ref{alg:high cost} in Appendix~\ref{sec:alg:high cost}).

\begin{definition}
For an $m$-bucket ordering $\sigma$, a vertex $u\in V$ and integer $i\in [m]$, let $\sigma_{u\rightarrow i}$ denote the bucket order
defined by leaving the value of $\sigma(v)$ unchanged for $v\neq u$ and mapping $u$ to $i$.  More precisely:
$\sigma_{u\rightarrow i}(v) = \sigma(v)$ if $v\neq u$, and
$\sigma_{u\rightarrow i}(u) = i$.
\end{definition}

Note that $\sigma_{u\rightarrow i}$ may not be exactly an $m$-bucket ordering.
To be precise, we will say  that $\sigma_{u\rightarrow i}$ is an  $m$-bucket$^*$ ordering whenever $\sigma$ is an $m$-bucket ordering, for every $u\in V$ and $i\in [m]$.
Clearly, Observation~\ref{obs:bucket} holds for $m$-bucket$^*$ orderings as well, with a possible different constant hiding in the $\Theta$-notation.
The following lemma is proven using standard measure concentration inequalities:

\begin{lemma}\label{crux}
Fix an $m$-bucket ordering $\sigma$ of $V$.  With probability at least $1-n^{-10}$, for all $u\in V$ and $i\in [m]$:
\begin{equation}\label{cruxeq} \left |  \cost^{u,S}(\sigma_{u\rightarrow i}) - \cost^u(\sigma_{u\rightarrow i}) \right | = O(\eps \gamma n)\ .
\end{equation}
\end{lemma}

\noindent
By  summing (\ref{cruxeq}) over all $(u,i)$ such that $i=\sigma(u)$, we get
\begin{corollary}\label{cor:main}
For any $m$-bucket order $\sigma$ with probability at least $1-n^{-10}$,
$ \left |  \cost^S(\sigma) - \cost(\sigma) \right | = O(\eps \gamma n^2)$.
\end{corollary}

\bignegspace
\subsubsection{Arora et al's LP approach \cite{DBLP:journals/mp/AroraFK02}}\label{sec:aroraLP}

The benefit of  Lemma~\ref{crux} is the fact that (\ref{cruxeq}) can be written as a pair of linear inequalities in  variables $(x_{vj})_{v\in V\setminus\{u\}, j\in [m]}$,
where $x_{vj}$ is indicator for the predicate ``$\sigma(v) = j$''.
Indeed, $\cost^{u,S}(\sigma_{u\rightarrow i})$ is a known constant, and $\cost^u(\sigma_{u\rightarrow i})$ is a linear
combination of $(x_{vj})_{v\neq u, j\in [m]}$.
This property allowed  Arora et. al in  \cite{DBLP:journals/mp/AroraFK02} to introduce an LP over these variables, where the utility function  $\cost^S(\sigma)$ is clearly a linear function of
the system $(x_{vj})_{v\in V, j\in [m]}$.  Some obvious standard constraints are added: 
For all $v,j$, $x_{vj}\geq 0$ and for all $v$, $\sum_{j\in [m]} x_{vj}=1$, and of course $x_{vj}$ is hardwired as $1$ (resp. $0$) whenever $v\in S$ and $\sigma^*(v)=j$ (resp. $\sigma^*(v)\neq j$). The  \emph{almost balanced bucket} constraint  is also added:
$ \forall j\in [m]: \lfloor n/m\rfloor \leq \sum_{v\in V} x_{vj} \leq \lceil n/m\rceil\ .$
The following arguments in \cite{DBLP:journals/mp/AroraFK02} are by now classic: Randomly round the optimal LP solution
by independently drawing, for each $v\in V$, from the discrete distribution
assigning probability $x^*_{vj}$ to the $j$'th bucket.
Denote the resulting $m$-bucket order $\sigma'$.
As argued in \cite{DBLP:journals/mp/AroraFK02}, with high probability each
constraint in the system will be satisfied up to a possible
additive violation of magnitude depending on an $\ell_\infty$
and an $\ell_0$ (support size) property of the constraint. The precise statement is as follows:

\begin{lemma}[Essentially \cite{DBLP:journals/mp/AroraFK02}] \label{lem:LP beta}
If the optimal solution to the LP $x^*$ satisfies $\sum \beta_{vj} x^*_{vj} \leq \alpha$ for  $\beta \in \R^{|V|\times m}$ and $\alpha \in \R$, then with probability at least $1-\eta$ the rounded solution $\sigma'$ will violate the constraint by no more than $\|\beta\|_\infty\sqrt{\|\beta\|_0 \log (1/\eta)}$, where $\|\beta\|_0$ is the number of vertices $v\in V$ such that $\beta_{vi} \neq 0$ for some $i\in [m]$,  $\|\beta\|_\infty = \max_{v\in V, i\in [m]} |\beta_{vi}|$ and $\eta>0$ is any number.
\footnote{We have implicitly viewed $\sigma'$ as a vector $(\sigma'_{vj})_{v\in V, j\in [m]}$, with $\sigma'_{vi}$ indicator for $\sigma'(v)=i$.}
\end{lemma}
 
In our case, consider an  LP constraint coming from (\ref{crux}).  Its corresponding coefficient vector $\beta$ satisfies  $\|\beta\|_0 \leq n$ and $\|\beta\|_\infty \leq 1$.
We conclude that  with probability at least $1-n^{-10}$, (\ref{cruxeq}) is satisfied with $\sigma = \sigma'$ and for all $u,i$,  and hence also the guarantee of Corollary~\ref{cor:main}.\footnote{All this happens with with possibly slightly worse constants hiding in the $O$-notation.}  Also, in virtue of the
\emph{almost balanced bucket} constraint and Lemma~\ref{lem:LP beta}, with probability at least $1-n^{-10}$ the rounded solution $\sigma'$ is
an $m$-bucket order.
  Additionally, by analyzing the coefficient vector $\beta_\utility$ corresponding to the LP utility function, with probability at least $1-n^{-10}$ the cost $\cost^S(\sigma')$ is bounded by $\LP(x^*) + O(\eps\gamma n^2)$, which
is bounded by $\cost^S(\sigma^*) + O(\eps \gamma n^2)$ by LP optimality.
Note also that the guarantee of Corollary~\ref{cor:main} also applies to $\sigma = \sigma^*$ with probability at least $1-n^{-10}$.  Combining using union bound and triangle inequality, one gets that $\cost(\sigma') \leq \cost(\sigma^*) + O(\eps\gamma n^2)$. 
We conclude the section with the following lemma that is implicit in \cite{DBLP:journals/mp/AroraFK02}
\begin{lemma} \label{lem:LP rounding works}
Given the correct bucketing on the vertices of $S$, one can construct a polynomially sized linear program whose rounded solution $\sigma'$ has the property $\cost(\sigma')\leq \cost(\sigma^*)+O(\gamma \eps n^2)$ with probability at least $1-n^{-9}$.
\end{lemma}

\negspace
\subsection{Query efficiency} \label{sec:high cost query PTAS}
\negspace
The problem is that expressing inequality (\ref{crux}) in the LP requires complete knowledge of  the input $T$.
If we take a revised look at this  strategy, we see that
the sample $S$ is not strong enough in the sense that it can be used to well approximate $\cost(\sigma)$ (per Corollary~\ref{cor:main})
for no more than $\poly(n)$ 
$m$-bucket orders $\sigma$ simultaneously, but certainly not for \emph{all} $m$-bucket orders.

For each $u\in V$ randomly select a sample $S^u  = (v^u_1,\dots, v^u_p)$ of vertices of $V$, where $p = O((\gamma\eps)^{-2} \log n)$, each sample
$S^u$ is chosen independently of the other samples, and the $v^u_i$'s are chosen uniformly at random from $V$, with repetitions.
Denote the ensemble $\{S^u: u\in V\}$ by $\S$.  For any $m$-balanced ordering $\pi$ on $V$, define
$\cost^{u,\S}(\pi) = \frac {n} {2p} \sum_{i=1}^p \cost^{u,v^u_i}(\pi)$ and
$\cost^\S(\pi) = \sum_{u\in V} \cost^{u,\S}(\pi)$.
It is not hard to see that $\cost^\S(\pi)$ is an unbiased estimator of $\cost(\pi)$, for any $\pi$.
Using standard measure concentration bounds, we have the following:
\begin{lemma}\label{thm:bigsample}
With probability at least $1-n^{-10}$, uniformly for all $m$-balanced orderings $\sigma$ on $V$,
$\left | \cost^\S(\sigma) - \cost(\sigma) \right | = O(\eps\gamma n^2)$.
\end{lemma}
\begin{lemma}\label{thm:bigsample1}
Fix an $m$-balanced ordering $\sigma$. With probability at least $1-n^{-10}$, uniformly for all $u\in V$ and $i\in [m]$,
\begin{equation}\label{eq:bigsample1}
\left | \cost^{u,S}(\sigma_{u\rightarrow i}) - \cost^{u,\S}(\sigma_{u\rightarrow i}) \right | = O(\eps\gamma n)\ .
\end{equation}
\end{lemma}
\noindent
By  summing  (\ref{eq:bigsample1})
over all $(u,i)$ s.t. $i=\sigma(u)$, we get
\begin{corollary}\label{thm:bigsample2}
Fix an $m$-balanced ordering $\sigma$. With probability at least $1-n^{-10}$,
$\left | \cost^{S}(\sigma) - \cost^\S(\sigma) \right | = O(\eps\gamma n^2)$.
\end{corollary}
We build an LP as in Section~\ref{sec:aroraLP}, except that (\ref{eq:bigsample1}) replaces (\ref{cruxeq}).
Note that the coefficient vectors $\beta$ of the new constraints
now satisfy $\|\beta\|_0 = O(p) = O((\gamma \eps)^{-2} \log n)$ and $\|\beta\|_\infty = O(n/p) = O(n\gamma^2\eps^2/\log n)$.
Using Lemma \ref{lem:LP beta} and a similar  analysis as in Section~\ref{sec:aroraLP}, an analog of Lemma \ref{lem:LP rounding works} can be proven. That is, we conclude that with probability at least $1-n^{-9}$,
the $m$-bucketed ordering $\sigma'$ outputted by rounding the optimal LP solution satisfies $ \cost^\S(\sigma') \leq \cost^\S(\sigma^*) + O(\eps \gamma n^2)$.  By Lemma~\ref{thm:bigsample} this implies that $\cost(\sigma') \leq \cost(\sigma^*) + O(\eps\gamma n^2)$.
Algorithm \ref{alg:high cost} (Appendix~\ref{sec:alg:high cost}) summarizes the query efficient PTAS for MFAST high cost case.
The $k$-CC high cost case can be solved in similar lines,
though this case is slightly easier because the clusters
need not be balanced.  
\negspace
\section{Discussion and Future Work}
We believe that in the low cost $k$-CC case, there should be a PTAS with efficient query complexity, running in time $\poly(n, \eps^{-1}, k)$ (not exponential in $k, \eps^{-1}$), assuming the low cost case in each recursive instance.  This is true for MFAST, and we leave the question of achieving it for $k$-CC to future work.

\negspace
\bibliographystyle{plain}
\bibliography{sublinear_PTAS}

\begin{thebibliography}{10}

\bibitem{Ailon11:active}
Nir Ailon.
\newblock An active learning algorithm for ranking from pairwise preferences
  with an almost optimal query complexity.
\newblock {\em Journal of Machine Learning Research}, 13:137--164, 2012.

\bibitem{AilonB12}
Nir Ailon and Ron Begleiter.
\newblock Active learning of custering with side information using smooth
  relative regret approximations.
\newblock {\em arXiv:1201.6462}, 2012.

\bibitem{Ailon:2008:AII}
Nir Ailon, Moses Charikar, and Alantha Newman.
\newblock Aggregating inconsistent information: {Ranking} and clustering.
\newblock {\em Journal of the ACM}, 55(5):23:1--23:27, October 2008.

\bibitem{Alon06}
Noga Alon.
\newblock Ranking tournaments.
\newblock {\em SIAM Journal on Discrete Mathematics}, 20, 2006.

\bibitem{DBLP:journals/mp/AroraFK02}
Sanjeev Arora, Alan~M. Frieze, and Haim Kaplan.
\newblock A new rounding procedure for the assignment problem with applications
  to dense graph arrangement problems.
\newblock {\em Math. Program.}, 92(1):1--36, 2002.

\bibitem{BBC04}
Nikhil Bansal, Avrim Blum, and Shuchi Chawla.
\newblock Correlation clustering.
\newblock {\em Machine Learning}, 56:89--113, 2004.

\bibitem{basu05}
Sugato Basu.
\newblock {\em Semi-supervised Clustering: {P}robabilistic Models, Algorithms
  and Experiments}.
\newblock PhD thesis, Department of Computer Sciences, University of Texas at
  Austin, 2005.

\bibitem{CharikarW04}
Moses Charikar and Anthony Wirth.
\newblock Maximizing quadratic programs: Extending grothendieck's inequality.
\newblock In {\em FOCS}, pages 54--60. IEEE Computer Society, 2004.

\bibitem{GiotisGuruswami06}
Ioannis Giotis and Venkatesan Guruswami.
\newblock Correlation clustering with a fixed number of clusters.
\newblock {\em Theory of Computing}, 2(1):249--266, 2006.

\bibitem{MatSch2007}
C.~Kenyon-Mathieu and W.~Schudy.
\newblock How to rank with few errors.
\newblock In {\em STOC '07: Proceedings of the thirty-ninth annual ACM
  Symposium on Theory of Computing}, pages 95--103. ACM, 2007.

\bibitem{DBLP:conf/soda/Swamy04}
Chaitanya Swamy.
\newblock Correlation clustering: maximizing agreements via semidefinite
  programming.
\newblock In {\em SODA}, pages 526--527, 2004.

\end{thebibliography}

\appendix

\section{Proof of Lemma~\ref{lem:deg in S and C equal}}\label{sec:proof:lem:deg in S and C equal}

This is a simple application of the following more general
well known sampling principle:  If $V_1,\dots, V_M$ is a collection of subsets 
of $V$ and $T $ is a sample
of $N$ uniformly chosen elements from $V$ (with repetition),
then with probablity $1-\eta$,
 for all $i=1,\dots, M$, 
$ \left | \frac { |V_i|}n - \frac {|V_i\cap T|}{|T|}\right | = O\left (\sqrt{N^{-1}\log (M/\eta)}\right )$.


\section{Proof of Lemma~\ref{lem:Cstar almost hat C}}\label{sec:proof:lem:Cstar almost hat C}

We start by proving the inclusion 
\begin{equation}\label{inclusion1} \hat C_j \subseteq C^*_j \cup \Vcostly\ .\end{equation} Assume for contradiction that there exist some $v \in \hat C_j \setminus (C^*_j \cup \Vcostly)$. Let $i \in [k]$, $i\neq j$ be such that $v \in C^*_i$. As $v \notin \Vcostly$ we know by Theorem \ref{thm:cost_add_apx} that 
$ \cost^*(v,i) = \cost^*(v) \leq \cost^*(v,j) \leq \tcost(v,j)+\beta n \leq  \frac{c_3 n}{k^2} + \beta n$. 
We get:
\begin{equation} \label{eq:C_i+C_j}
\frac{2nc_3}{k^2} + 2n\beta \geq \cost^*(v,i)+\cost^*(v,j) \geq {|C^*_i|+|C^*_j|-1}\ ,
\end{equation}
where the right hand inequality is a consequence of the fact
for any $u\in (C^*_i\cup C^*_j)\setminus \{v\}$, $v$ will either
incur a price w.r.t. $u$ if it is included in $C^*_i$ or it is included
in $C^*_j$.
Hence,
\begin{equation}\label{maxCiCj}
 \max\{ |C^*_i|,|C^*_j| \}\leq \frac{2nc_3}{k^2} + 2n\beta + 1 \leq \frac{c_5n}{k^2} 
\end{equation}
for some constant $c_5$ that can be made arbitrarily small by tuning $c_2$ (the constant product in $\beta$) and $c_3$.
Since this holds for any $i \neq j$ satisfying  $C^*_i \cap (\hat C_j \setminus \Vcostly) \neq \emptyset$ we have that 
$$ |C^*_j| \geq |\hat C_j| -|\Vcostly| - \sum_{i :\; i \neq j , C^*_i \cap (\hat C_j \setminus \Vcostly) \neq \emptyset} |C^*_i|  \geq 
\frac{n}{2k} - \frac{c_4n}{k} - \frac{c_5n}{k} > \frac{c_5n}{k}\ ,$$
where we used (\ref{Vcostlysize}) and (\ref{maxCiCj}),
and ensure that   $2c_5+c_4<1/2$. We derive a contradiction to
(\ref{maxCiCj}).

We now prove that the inclusion $C^*_j \subseteq \hat C_j \cup \Vcostly$.  Notice that by (\ref{inclusion1}), we conclude that 
$|C^*_j|$ is lower bounded by $\frac n k \left (\frac 1 {2} - {c_4} \right ) \geq \frac n {4k}$,
 as long as $c_4 < 1/4$.  
Assume for the sake of contradiction that there exists some $v \in C^*_j \setminus \Vcostly$ such that for some $i\neq j$, $v \in \hat C_i$.
By the guarantee of Theorem~\ref{thm:cost_add_apx} we have that $\cost^*(v,j) \leq \cost^*(v,i) \leq \tcost(v,i) + \beta n\leq c_3 n/k^2 + \beta n$.  This gives us again (\ref{eq:C_i+C_j}), leading to (\ref{maxCiCj}), contradicting our lower bound on $|C^*_j|$
for sufficiently small $c_5$.
 
\noindent
This concludes the lemma proof.

\section{Continuation of Proof of Lemma~\ref{lem:close in V}}\label{sec:proof:lem:close in V}
We now proceed with the proof of the lemma.

Consider the following bipartite directed graph $H = (U, \Gamma)$.  The vertex set $U$ is defined as follows: 
$ U = \{C:\ C\mbox{ is a cluster in }\S\} \cup \{D: D\mbox{ is a cluster in }\tilde \S\}$.
The edge set $\Gamma$ is defined as follows:  For any
cluster $C \in \S$, add a directed edge $(C,D)$, where  $D$ maximizes $|D' \cap C|$ over clusters $D'$ of $\tilde \S$ breaking ties arbitrarily.  
Symmetrically,  for cluster $D$ of $\tilde S$ add a directed
edge $(D,C)$ where $C$ maximizes $|C'\cap D|$ over clusters
$C'$ of $\S$, breaking ties arbitrarily.
Note that the out degree of all vertices of $H$ is exactly $1$.
We will now define a bipartite matching on $U$ using
the following rule:
If for some $C,D$, both $(C,D)\in \Gamma$ and $(D,C)\in \Gamma$, then match $C$ to $D$, and call the pair $(C,D)$ \emph{a good match}.  The remaining (unmatched) vertices are matched arbitrarily , and the corresponding pairs are called  \emph{bad matches}.

By the above claim, if $(C,D)$ is a good match, then
$ \max\{|C\setminus D|, |D\setminus C|\} = O(\gamma^{1/3}|S|)$.

We now show that if $(C,D)$ is a bad match then both
$|C| = O(\delta^{1/3}|S|)$ and $|D| = O(\delta^{1/3}|S|)$.
By symmety, it suffices to show that $|C|=O(\delta^{1/3}|S|)$.
Let $C$ be a cluster of $\S$ that is a member of a bad match. Let $D$ be the unique cluster of $\tilde \S$ such that $(C,D)\in\Gamma$ and let $C'$ be the unique cluster of $\S$ such that $(D,C')\in\Gamma$. By the definition of a bad match we know that $C \neq C'$. By the above claim we have that both $|C\setminus D| = O(\delta^{1/3}|S|)$
and $|D\setminus C'| = O(\delta^{1/3}|S|)$, which implies
$|C\setminus C'| = O(\delta^{1/3}|S|)$.  But $C\cap C'=\emptyset$, therefore $|C| = O(\delta^{1/3}|S|)$.
This concludes the proof of the lemma.

\section{Algorithm for High-Cost MFAST}\label{sec:alg:high cost}
\begin{algorithm}
\caption{query efficient PTAS for MFAST}
 \label{alg:high cost}
\emph{Input}: A graph T=(V,A) an approximation parameter $\eps$ and assumed minimal cost parameter $\gamma$

\emph{Output}: a permutation $\sigma: V \to [n]$ where $n=|V|$

For each $u \in V$ randomly select a sample $S^u=\{v_1^u,\ldots,v_p^u\}$, where $p=O( (\eps \gamma)^{-2} \log(n) )$ and the $v_i^u$'s are chosen from $V$ with repetitions.  Denote the ensemble $\{S^u: u\in V\}$ by $\S$. (This is the \emph{verification} sample.)

Set $S$ as a set of random i.i.d.\ vertices $S=\{v_1,\ldots,v_s\}$ chosen with repetitions, where $s=O( (\eps \gamma)^{-2} \log(n) )$.  (This is the \emph{enumeration} sample.)

Set $m=O((\gamma \eps)^{-1})$ as the number of buckets.

For each possible $m$-bucket order of the vertices of $S$ perform the following:
\begin{itemize}
\item Construct an LP as described in Section~\ref{sec:high cost query PTAS}, producing a fractional $m$-bucket order that agrees with the bucketing of  $S$.

\item Solve the LP and round it as described in Section \ref{sec:aroraLP}.
\end{itemize}
Pick the rounded solution whose approximated cost w.r.t.\ $\cal S$ ($\cost^{\cal S}(\cdot)$) is minimal, and output an arbitrary permutation extending it.
\end{algorithm}

\end{document}